\documentclass[twoside,letterpaper,11pt]{article}
\usepackage{amsmath}
\usepackage{amsthm}
\usepackage{amssymb}
\usepackage{amsfonts}
\usepackage{geometry}
\usepackage{fullpage}
\usepackage{graphicx,color}
\usepackage[dvipsnames]{xcolor}
\usepackage{fancybox}
\usepackage{tikz}
\usepackage{subcaption}
\usepackage{hyperref}
\usepackage{mleftright}
\usetikzlibrary{positioning,chains,fit,shapes,calc}
\usetikzlibrary{trees}
\usetikzlibrary{decorations.pathmorphing}
\usetikzlibrary{decorations.markings}

\usepackage{cool}
\Style{DSymb={\mathrm d},DShorten=true,IntegrateDifferentialDSymb=\mathrm{d}}

\newtheorem{theorem}{Theorem}

\newtheorem{lemma}[theorem]{Lemma}
\newtheorem{proposition}[theorem]{Proposition}
\newtheorem{claim}[theorem]{Claim}

\numberwithin{theorem}{section}

\newcommand{\abs}[1]{\left\vert#1\right\vert}
\newcommand{\set}[1]{\left\{#1\right\}}
\newcommand{\eps}{\varepsilon}

\newcommand{\CommentS}[1]{}

\newcommand{\pr}{\mathbb{P}}

\title{FPTAS for \#BIS with Degree Bounds on One Side }

\author{ Jingcheng Liu
        \thanks{Computer Science Division, UC Berkeley. Email: {\tt liuexp@berkeley.edu}. Supported in part by NSF grants CCF-1420934 and CCF-1343104.}
        \and
        Pinyan Lu\thanks{Microsoft Research. Email: {\tt pinyanl@microsoft.com}. Part of this work was done while the author was visiting the Simons Institute for the Theory of Computing, Berkeley}
}
\date{}
\begin{document}

\maketitle

\begin{abstract}
Counting the number of independent sets for a bipartite graph (\#BIS) plays a crucial role in the study of approximate counting.
It has been conjectured that there is no fully polynomial-time (randomized) approximation scheme (FPTAS/FPRAS) for \#BIS, and it was proved that the problem for instances with a maximum degree of $6$ is already as hard as the general problem.
In this paper, we obtain a surprising tractability result for a family of  \#BIS instances. We design a very simple deterministic fully polynomial-time approximation scheme  (FPTAS) for \#BIS when the maximum degree for one side is no larger than $5$. There is no restriction for the degrees on the other side, which do not  even have to be bounded by a constant. Previously, FPTAS was only known for instances with a maximum degree of $5$ for both sides.
\end{abstract}

%

\section{Introduction}
Counting the number of independent sets in a bipartite graph (\#BIS) is arguably the most important open question in the study of approximation algorithms for counting problems, which plays a similar role as the unique game for optimization problems, or the PPAD class for fixed points and Nash equilibria. We do not know if it admits a fully polynomial-time (randomized) approximation scheme (FPTAS/FPRAS), and we do not know if it is as hard as counting the satisfying assignments for a satisfaction problem (\#SAT) either. It is conjectured to be of intermediate complexity~\cite{dyer2000relative}. Similar to unique game, the approximability of \#BIS is important not only because it is an interesting problem on its own, but mainly due to the fact that many other counting problems are proved to have the same complexity as \#BIS. It is a complete problem for a family of logically defined problems called $\#{\rm RH}\Pi_1$ as a subfamily of \#P~\cite{dyer2000relative}. With the help of this intermediate class, a number of complete classifications for the approximability for various families of problems have been proved, such as the Boolean \#CSP problems~\cite{dyer2010approximation,DyerGJR12,BulatovDGJM13}.

Without restricting input graphs to be bipartite, the approximability for counting the number of independent sets (\#IS) is well understood. For general graphs, approximately counting the number of independent sets is as hard as finding the maximum independent set, which is NP-hard.
This reduction was one of the very first proofs for inapproximability for counting problems.
The hard instances used in the reduction have very large degrees, and as a result later research has been mainly focused on sparse graphs, such as graphs with a maximum degree bound. An FPRAS based on the Markov chain Monte Carlo (MCMC) method was obtained when the maximum degree is $3$ in~\cite{dyer2000markov} and then $4$ in~\cite{IS_LV97}. Later, a deterministic FPTAS based on the correlation decay technique was obtained for graphs with a maximum degree of $5$ by Weitz~\cite{Weitz06}. On the inapproximability side, it was proved that the problem is NP-hard as long as we allow the maximum degree to be $25$~\cite{IS_DFJ02}. The hardness bound was eventually reduced to $6$ and thus closed the gap in~\cite{Sly10}.

However, the approximability for \#BIS is much more challenging. We do not know any NP-hardness result even if we do not have a degree bound. The previous proof for general graphs does not work because finding a maximum independent set for bipartite graph is equivalent to finding a maximum matching (Konig's theorem), which is polynomial time solvable rather than being NP-hard.
The main reason to make \#BIS extremely important in the study of approximate counting is that  a large number of other problems are proved to have the same complexity as \#BIS (\#BIS-equivalent) or at least as hard as \#BIS  (\#BIS-hard) under approximation-preserving reduction (AP-reduction)~\cite{dyer2000relative}. Examples include combinatorial counting problems such as  \#Downsets (counting the number of downsets of a partial order system), \#Bipartite-$q$-COL, \#Bipartite-MAX-IS (all in~\cite{dyer2000relative}) and \#Stable-Matching~\cite{chebolu2012complexity}, logical problems such as \#1P1NSAT and \#IM~\cite{dyer2000relative}, problems from statistical physics such as
computing the partition problems for ferromagnetic Ising model with mixed external fields~\cite{goldberg2007complexity} and Potts system~\cite{GoldbergJ12}, and many other counting problems. One recent interesting result on \#BIS itself indicates that  \#BIS with maximum degree $6$ is already as hard as general \#BIS~\cite{cai2014bis}. This restricted version of \#BIS is more useful in some reductions and the new result has been used to prove  \#BIS-hardness for other problems such as ferromagnetic two-spin systems with a uniform external field~\cite{LLZ14}.
Moreover, it was shown that if \#BIS does not admit an FPRAS, then there is an infinite approximation hierarchy even within \#BIS~\cite{bordewich2011approximation}.

The main reason to make \#BIS flexible in these reductions is indeed due to its bipartite structure, on which the vertices from two sides can encode (or be encoded by) two different objects for other problems.
For example, a hypergraph can be represented as a bipartite graph (known as its incidence graph), with the left side being the vertex set and the right side being the edge set.
In this new bipartite graph, the degrees on the left side are the same as the
degrees in the hypergraph, while the degrees of the right side are sizes of hyperedges
in the hypergraph. This nature makes it suitable to study \#BIS with different
degree constraints on two sides. For example, \#Semi-regular-BIS studied
in~\cite{GoldbergJ12} has one side regularity requirement.

On the algorithmic side,
it was shown in \cite{IS_DFJ02,inapp_MWW09} that any local MCMC algorithm that uses subsets of vertices as state space, mixes slowly even on a bipartite graph with a maximum degree of $6$.
More recently, an interesting attempted Markov chain by Ge and Stefankovic~\cite{ge2012graph}, which uses subsets of edges as state space and differs from previous MCMC methods, was also shown to mix slowly in~\cite{goldberg2012counterexample}.
Prior to our work, the best known FPRAS or FPTAS for \#BIS was the same as that for \#IS for graphs with a maximum degree of $5$. There was no algorithmic evidence to distinguish \#BIS from \#IS.

\subsection*{Our Results}
Our main result is an FPTAS for \#BIS when the maximum degree for one side is no larger than $5$. There is no restriction for the degrees on the other side, which do not even have to be bounded by a constant. Assuming that there is no FPTAS or FPRAS for general \#BIS, our result is of the best possible in the sense that, if we allow degrees of $6$ on both sides, the problem is already \#BIS-hard. Our FPTAS can also be viewed as the first algorithmic evidence to  distinguish \#BIS from \#IS.

Our algorithm is almost identical to Weitz's algorithm for general \#IS with a maximum degree of $5$, and the main technique is also
correlation decay. We elaborate a bit on the ideas.  Due to a standard argument, computing the number of independent sets is reduced to computing the marginal probability of a vertex to be chosen, if one samples an independent set uniformly at random from all possible independent sets of the input graph. Then, the main idea is to estimate these marginal probabilities directly rather than through sampling,  which is made possible by the remarkable \emph{self-avoiding walk} (SAW) tree introduced by Weitz in~\cite{Weitz06}. For efficiency of computation, the marginal distribution of a vertex is estimated using only a local neighborhood around a vertex. To justify the precision of the estimation, we show that far-away vertices have little influence on the marginal distribution. This is done by analyzing the decay rate of correlation between two vertices in terms of their distance. In~\cite{Weitz06}, it is proved that when the degree of each vertex is at most $5$, this decay rate is exponentially small in the depth of the SAW tree.
However, the same analysis does not apply to our case as the degrees of one side can be arbitrarily large.
To overcome this, our main idea is to combine two recursion steps of the SAW for \#BIS into one, and work with this two-layer recursion instead.
As it turns out, it has the same effect as treating one side of vertices as variables, while the other side of vertices as constraints.
Then we ensure that the degrees in the first layer, which are the variables' degree, are always no more than $5$.
The key is to formalize an observation that the larger the second layer degree (the constraint's degree), the faster the correlation decays.
Such analysis is only possible for \#BIS rather than general \#IS.

Such a two-layer type recursion is similar to that for monotone CNF and hypergraph matching in~\cite{monotone-cnf}. As we have seen there, the analysis for these two-layer recursions is usually much more challenging and complicated. One additional complication here is due to the fact that the degrees for the other side are not even bounded by a constant. For these cases, we need to prove an even stronger notion of correlation decay called \emph{computationally efficient correlation decay} as in~\cite{LLY12,counting-edge-cover,monotone-cnf}, which says that the error decays by a super-constant factor if we go through a vertex with a super-constant degree. In order to prove the correlation decay property, we use a potential function to amortize the decay rate as in many previous works~\cite{RestrepoSTVY11,LLY12,LLY13,SST,liu2014fptas}. A good potential function is the key to these proofs. In this paper, the potential function is carefully constructed to not only make the decay rate less than one but also make the proof simpler.
Effectively, the potential function we use in this paper makes the amortized decay rate of the two-layer recursion act as if it is a single layer. This dramatically simplifies the proof. We believe this simple idea can find applications in the analysis of other two-layer recursions.

\subsection*{Related work}
The correlation decay based FPTAS for counting independent sets was extended to anti-ferromagnetic two-spin systems~\cite{LLY12,SST,LLY13}. From a statistical physics point of view, the independent set problem is a special case of the hard-core model, where one introduces an \emph{activity parameter} and counts weighted independent sets. To extend our result to weighted independent sets and anti-ferromagnetic two-spin systems in general is an interesting open question.

There are some other works that study counting problems for richer families of graphs other than a single maximum degree constraint. A beautiful direction is to replace the maximum degree constraint by the connective constant~\cite{sinclair2013spatial,sinclair2014spatial}, which can be viewed as a version of average degree. However, if one would like to apply this average degree type argument to the \#BIS instances in our setting, the connective constant is unbounded since the degree of one side is unbounded.
Our result also indicates that in the case of bipartite graphs, the average degree may not be powerful enough to capture the complexity of the problem.

Such phenomena where larger degrees (the degrees of constraints) only make the problem easier, also come up in hypergraph independent sets.
In particular, let $d$ be the maximum degree, and $m$ be the minimum edge size (which plays the role of constraint degree).
As is shown in~\cite{bordewich2006},  if $m\ge d+2 \ge 5$, the problem of counting independent sets on such hypergraphs admits FPRAS.
In contrast, if we only have maximum degree parameter $d$, then it only admits FPTAS when $d\le 5$~\cite{monotone-cnf}.

Bipartiteness changing the complexity of a problem is also an interesting phenomenon in the study of approximate counting.
Two other famous examples are graph colorings and perfect matchings.
Counting the number of colorings for bipartite graphs is an important open question, which is known to be \#BIS-hard, but not known to be  \#BIS-equivalent or not.
There is an FPRAS for counting perfect matchings in a bipartite graph~\cite{app_JSV04}, while for general graphs it is a long-standing open question.

\section{Preliminary}
For an undirected graph $G=(V,E)$, a subset of vertices $I\subseteq V$ is an \emph{independent set} of $G$ if there is no edge between any two vertices within $I$. We denote $I(G)$ as the set of independent sets of graph $G$, and $Z(G)\triangleq \abs{I(G)}$.
$G$ is \emph{bipartite} if there exists $U \subseteq V$ such that both $U \in I(G)$ and $U^{C} \in I(G)$. Hence it can be written as $G=(U \uplus U^{C}, E)$.


Given a graph $G=(V,E)$, a vertex $u \in V$, a set of vertices $U \subseteq V$, we define the following:
\begin{itemize}
	\item Removing a vertex $u$ and its incident edges:
		\[ G- u = \left(V\setminus \set{u}, \set{e\in E \mid e \text{ is not incident with } u}\right).\]
	\item Removing a set of vertices $U$ and all incident edges:
		\begin{align*}
			G- U = &\left(V\setminus U,\right.\\
			&\left.\set{e\in E\mid e \text{ is not incident with any } u \in U}\right).
		\end{align*}
\end{itemize}
We write $N_G(u)$ for an open neighborhood of a vertex $u$ (which does not contain $u$), and $N_G[u]$ for a closed neighborhood of $u$ (which includes $u$ itself).
Note that in the case of a bipartite graph $G=(U \uplus V, E)$,
for every vertex $u \in U$, we have $N_G(u) \subseteq V$.

In general, we use $u=0$ to refer the vertex $u$ is not chosen in an independent set, and $u=1$ for being chosen.
With an independent set sampled uniformly at random, the probability that the vertex $u$ is chosen is denoted by $\pr_G(u=1)$. Similarly,  $\pr_G(u=0)$ is for the probability that the vertex $u$ is not chosen.

As an easy observation, the number of independent sets without choosing $u$ is $Z(G-u)$,
with those choosing $u$ being $Z\left(G-N_G[u]\right)$.
Thus $\pr_G(u=0) = \frac{Z(G-u)}{Z(G)}$, and $\pr_G(u=1)= \frac{Z\left(G-N_G[u]\right)}{Z(G)}$.


\section{The Algorithm}
The main result of this paper is the following algorithm.
%

\begin{theorem}\label{thm:bis}
	There is an FPTAS for counting the number of independent sets of a bipartite graph $G=(U \uplus V, E)$ with
$\min\set{\Delta_U,\Delta_V}\leq 5$, where $\Delta_U$ and $\Delta_V$ are the maximum degree over vertex set $U$ and $V$ respectively.
\end{theorem}

Without loss of generality, we assume $\Delta_U \le \Delta_V$. Thus, we have  $\Delta_U \leq 5$. We denote $n = \abs{U}, m= \abs{V}$.



\subsection{Counting from Likelihood Ratios}
We shall first reduce the problem of counting to computing likelihood ratios. This is a standard reduction, and was introduced as the \emph{self-reducibility} structure in~\cite{jerrum1986random}.

For vertices $u \in U$, let $R(G, u) \triangleq \frac{\pr_G(u=1)}{\pr_G(u=0)} = \frac{Z(G-N_G[u])}{Z(G-u)}$.
Although we can similarly define $R(G,v)$ for $v \in V$, our ultimate algorithm would only involve vertices $u \in U$ as variables.

Let $u_1, u_2, \ldots, u_n$ be an arbitrary enumeration of vertices in $U$, and $G_i=G-\set{u_1,\ldots,u_{i-1}}$.
In particular, $G_1=G$, and $G_n - u_n = G - U$. Also recall that the vertex set of $G-U$ is just $V$, which is an independent set of $G$ by itself,
\begin{align*}
	Z(G)= & Z(G-u_1) + Z(G-N_G[u_1]) \nonumber \\
	=&Z(G_1-u_1) \cdot \left(1+R(G_1,u_1)\right) \nonumber \\
	=&\left(Z\left(G_2 - u_2\right) + Z\left(G_2-N_{G_2}\left[u_2\right]\right)\right) \cdot \left(1+R(G_1,u_1)\right) \nonumber \\
	=&Z(G_2 - u_2) \cdot \left(1+R(G_2,u_2)\right)\cdot \left(1+R(G_1,u_1)\right) \nonumber \\
	&\vdots \nonumber \\
	=&Z(G_n - u_n)\cdot \prod_{i=1}^n \left(1+R(G_i,u_i)\right) \nonumber \\
	=&Z(G - U)\cdot \prod_{i=1}^n \left(1+R(G_i,u_i)\right) \nonumber \\
	=&2^m \prod_{i=1}^n \left(1+R(G_i,u_i)\right) \nonumber \\
\end{align*}

\begin{proposition}
	\label{prop:oracle}
	Provided an algorithm $R(G, u, \eps)$ for estimating $R(G,u)$ within an additive error $\eps$, which runs in time $poly(n,1/\eps)$, and outputs $\hat{R}$ such that $\abs{\hat{R} - R(G,u)} \le \eps$.
	There is an FPTAS for estimating $Z(G)$ based on $R(G, u, \eps)$.
\end{proposition}
\begin{proof}
	Let $G_i=G-\set{u_1,\ldots,u_{i-1}}$.
	Given $0<\eps<1$, let $\hat{R_i} \triangleq R\left(G_i, u_i,\frac{\eps}{2n}\right)$ and  $R_i \triangleq R\left(G_i, u_i\right) $. 
	Consider the algorithm that returns $\hat{Z}(G) = 2^m \prod_{i=1}^n \left( 1+\hat{R_i} \right)$ as an approximation for $Z(G)= 2^m \prod_{i=1}^n \left( 1+R_i \right)$.  We have
	\begin{align*}
		&\frac{\abs{\hat{R_i} - R_i}}{1+R_i} \le \abs{\hat{R_i} - R_i} \le \frac{\eps}{2n} \\
		\implies & 		\left( 1 - \frac{\eps}{2n} \right) \le
		\frac{1+\hat{R_i}}{1+R_i} \le
		\left( 1 + \frac{\eps}{2n} \right).
	\end{align*}
	Since $\frac{\hat{Z}(G)}{Z(G)} = \prod_{i=1}^n \frac{1+\hat{R_i}}{1+R_i}$, we have,
	\begin{align*}
		\left( 1 - \frac{\eps}{2n} \right)^n \le&
		\prod_{i=1}^n \frac{1+\hat{R_i}}{1+R_i} =\frac{\hat{Z}(G)}{Z(G)}\le
		\left( 1 + \frac{\eps}{2n} \right)^n \\
		\implies& 1-\eps \le \frac{\hat{Z}(G)}{Z(G)} \le 1+ \eps.
	\end{align*}
This concludes the proof.
\end{proof}

Therefore, the remaining task is to design an algorithm for $R(G, u, \eps)$.

\subsection{Tree Recursion from Self-Reducibility}
Before implementing the algorithm required by Proposition \ref{prop:oracle},
we will show a recursive relation for $R(G,u)$ using the self-reducibility structure again, which gives an alternative derivation of Weitz's self-avoiding walk tree approach~\cite{Weitz06}.


\begin{lemma}
Let $d \triangleq \deg_G(u)$, and $N_G(u)$ be enumerated as $\set{v_i}_{i=1}^d$.
Denote  $G_i \triangleq (G - u) - \set{v_j}_{j=1}^{i-1}$, $w_i \triangleq \deg_{G_i}(v_i) $.
Let $N_{G_i}(v_i)$ be enumerated as $\set{u_{i,j}}_{j=1}^{w_i}$, and $G_{i,j} \triangleq (G_i - v_i) - \set{u_{i,k}}_{k=1}^{j-1}$. Then
	\[
		R(G, u) = \prod_{i=1}^d \left(1+ \prod_{j=1}^{w_i} \left(1+R\left(G_{i,j},u_{i,j}\right)  \right)^{-1}\right)^{-1}.
	\]
	\label{lem:rec}
\end{lemma}

We refer to $d$ as the first-layer degree, and $w_i$ as the second-layer degrees. If $d=0$ or $w_i=0$ for some $i$, we follow the convention that an empty product is $1$.

The same recursion can be obtained by first constructing the self-avoiding walk tree for $G$ from $u$ and then combining two steps of the tree recursion at a time.
Instead of explicitly constructing the whole SAW tree, we present an alternative derivation based only on a nontrivial partition scheme promised by the self-reducibility.

\begin{proof}
	Recall that $d \triangleq \deg_G(u)$, and $N_G(u)$ is enumerated as $\set{v_i}_{i=1}^d$, and $G_i - v_i = G_{i+1}$,

	\begin{align*}
		R(G,u) =& \frac{\pr_G(u=1)}{\pr_G(u=0)} 
		= \frac{\frac{Z(G-N_G[u])}{Z(G)}}{\frac{Z(G-u)}{Z(G)}}
		= \prod_{i=1}^d \frac{Z(G_i - v_i)}{Z(G_i)}
	\end{align*}
	Next we use the self-reducibility structure of the problem, which gives the following partition scheme for free: $Z(G_i) =Z(G_i-v_i)+Z(G_i - N_{G_i}[v_i]) $, thus
	\begin{align*}
		\frac{Z(G_i - v_i)}{Z(G_i)}
		&=  \frac{Z(G_i - v_i)}{Z(G_i-v_i)+Z(G_i - N_{G_i}[v_i])}\\
		&=  \frac{1}{1+\frac{Z(G_i - N_{G_i}[v_i])}{Z(G_i - v_i)}}\\
		&=  \frac{1}{1 + R(G_i, v_i)}.
	\end{align*}

	Similarly one could use self-reducibility again 
	and show that $R(G_i, v_i) = \prod_{j=1}^{w_i} \left( 1 +R(G_{i,j}, u_{i,j})  \right)^{-1}$. Substituting these $R(G_i,v_i)$ into the above recursion, we conclude the proof.
\end{proof}

It is worth noting that $\pr_{G_i} \left( v_i = 0 \right) = \frac{Z(G_i - v_i)}{Z(G_i)}$.
As an intuition, if $\pr_{G_i} \left( v_i = 0 \right) \approx 1$, namely $Z(G_i - v_i) \approx Z(G_i)$, then one could safely ignore the vertex $v_i$ and still get a good approximation.
We will see a more quantitative version of this fact, and in particular how it relates to the one-sided maximum degree in Claim \ref{prop:si}.

Since $0 < \frac{1}{\prod_{j=1}^{w_i} \left(1+R\left(G_{i,j},u_{i,j}\right) \right)} \le 1$, we can get the  following bound for $R(G,u)$ from the recursion, which will be useful in the analysis:
\begin{lemma}
	\label{rem:rgu}
	\[2^{-\deg_G(u)} \le R(G,u) \le 1.\]
\end{lemma}

We can further expand $R(G_{i,j}, u_{i,j})$s by the above recursion and get a tree recursion for $R(G,u)$. Since  $\Delta_U \le 5$,
except for the root of the recursion, we always have $d = \deg_{G_{i,j}}(u_{i,j})  \le \Delta_U - 1 \le 4$ for the first-layer degree.
In these cases, the above bounds are $\frac{1}{16} \le R(G,u) \le 1$.
%

\subsection{Computation Tree}
Now we are ready to implement the algorithm as required by Proposition~\ref{prop:oracle}.
We recursively define $R(G, u, L)$ as follows. 
For base case $L=0$, $R(G,u,L)=2^{-\deg_G(u)}$.
For $L>0$, let $L_i' = \max\left(0, L - \lceil \log_{45} (w_i+1) \rceil \right)$,
then 
\[
	R(G,u,L) = 	\prod_{i=1}^d \left(1+  \prod_{j=1}^{w_i} \left(1+R\left(G_{i,j},u_{i,j}, L_i'  \right)\right)^{-1}  \right)^{-1}.
\]

The recursion depth $L$ is used to control the accuracy of the estimation, and plays the same role as $\eps$ referred in Proposition~\ref{prop:oracle}.
After one step of recursion, $L$ is subtracted by $\lceil \log_{45} (w_i+1) \rceil$ rather than $1$, which is known as $M$-based depth introduced in~\cite{LLY12}, with $M=45$ in our case.

As an intuition, the recursion depth $L$ can be thought of as a computational budget, in which we replace every node with a branching degree greater than $45$ with a $45$-ary branching subtree.
Then, it is clear that the size of this branching computation tree up to depth $L$ is at most $O((45d)^L) = O(180^L)$,
and for second-to-base-case nodes (that is, nodes with $0<L \leq \lceil \log_{45}{(w+1)} \rceil $ ) they involve at most $O(n)$ extra base cases,
so the running time for the algorithm to compute $R(G,u,L)$ is $O(n 180^L)$.

By definition, our estimation $R(G,u,L)$ has the same bounds as $R(G,u)$ in Lemma~\ref{rem:rgu}.
\begin{lemma}
	\label{rem:rgl}
	\[2^{-\deg_G(u)} \le R(G,u,L) \le 1.\]
\end{lemma}

Formally we have the following key lemma, for which the proof is laid out in Section \ref{sec:cd}.
\begin{lemma}[Correlation Decay]
	Let $\alpha = 0.9616$. If $G=(U \uplus V, E)$ is a bipartite graph with $\Delta_U\le 5$, then for any $u \in U$,
	\begin{align}
		\abs{R(G, u, L) - R(G,u)} \le O( \alpha^L).
		\label{eqn:cd}
	\end{align}
	\label{lem:cd}
\end{lemma}

With this lemma, it is easy to estimate $R(G,u)$ by $R(G, u, L)$ with an additive error of $\epsilon$ by choosing $L=O(\log \frac{1}{\epsilon})$. Then combined with  Proposition \ref{prop:oracle}, we get the proof for Theorem \ref{thm:bis}.

\section{Analysis and Correlation Decay}
\label{sec:cd}
In this section, we establish the key correlation decay Lemma~\ref{lem:cd}.
To do that, a natural approach is to use induction, and show that the error decreases by a constant factor along each recursion step.
Unfortunately, this step-wise decay is not true in our case.
Instead, we perform an amortized analysis on the decay rate by a potential function, and show that step-wise decay is recovered on the new domain under the potential function.

In Section~\ref{sec:cd1}, we outline the induction and give a derivation of the \emph{amortized decay rate}.
We show that it suffices to bound the amortized decay rates as in Claim~\ref{prop:kappa} and \ref{prop:kappa2}, which are proved in Section~\ref{sec:cd2}.
In particular, we show how our choice of the potential function simplifies the amortized decay rates and the proof.

\subsection{Amortized Decay Rates}\label{sec:cd1}
We use
$\varphi(x) = \ln\left( \ln (1+x) \right)$
to map the values $R(G,u,L)$ and $R(G,u)$ into a new domain,
and prove the following:
\begin{align}
	\label{eqn:amortized-cd-all}
	\abs{\varphi \circ R(G,u,L) - \varphi \circ R(G,u)} \le 12\alpha^L.
\end{align}
The choice of this potential function will become clear in the next subsection.

\begin{claim}
	The condition (\ref{eqn:amortized-cd-all}) implies (\ref{eqn:cd}).
\end{claim}

\begin{proof}
	Note that $\varphi(x)$ is an increasing function.
	Let $R=R(G,u), \hat{R} = R(G,u,L)$,
recall the bounds from Lemma \ref{rem:rgu} and \ref{rem:rgl}, we have
\[
	\varphi(\frac{1}{2^{5}}) \le
	\varphi(R), \varphi(\hat{R}) \le
	\varphi(1).
\]

As a result, by Mean Value Theorem, $\exists \bar{y}: \varphi(\frac{1}{32})\le \bar{y} \le \varphi(1)$ such that
\begin{align*}
	\abs{\hat{R} - R} =& \D{\varphi^{-1}(y)}{y}\Big|_{y=\bar{y}} \cdot \abs{\varphi(\hat{R}) - \varphi(R)}\\
	\overset{(\clubsuit)}{\le}& 2\ln(2)\cdot 12\alpha^L=24\ln(2)\cdot \alpha^L.
\end{align*}
where $(\clubsuit)$ follows from the fact 
\[\D{\varphi^{-1}(y)}{y} = \left(1+\varphi^{-1}(y)\right) \ln \left(1+\varphi^{-1}(y)\right) \le 2\ln 2,\]
together with condition (\ref{eqn:amortized-cd-all}). This completes the proof.
\end{proof}

Since the case $d=5$ is applied  only once at the root, 
we first assume that $1\le d \le 4$ and show the following:
\begin{align}
	\label{eqn:amortized-cd}
	\abs{\varphi \circ R(G,u,L) - \varphi \circ R(G,u)} \le 4\alpha^L.
\end{align}
We prove it by induction on $L$.
Let $R=R(G,u), \hat{R} = R(G,u,L)$,
For the base case $L=0$,
we have
\[\abs{\varphi(\hat{R}) - \varphi(R)} \le \varphi(1) - \varphi(\frac{1}{32}) < 4. \]

Supposing the induction hypothesis holds for $L<l$, we prove that it also holds for $L=l$.
If $u$ is an isolated vertex, i.e. $d=0$, $\hat{R} = R = 1$. Now we assume that $d>0$.

Let
\[h(\mathbf{x}) = \prod_{i=1}^d \left(1+ \prod_{j=1}^{w_i} \left(1+x_{i,j}\right)^{-1} \right)^{-1},\]
which is the analytic version of the recursion.
We also write $h = h(\mathbf{x})$ for short.

Let $\mathbf{y}$ be the accurate vector with 
\[y_{i,j} = \varphi \circ R\left(G_{i,j},\ u_{i,j}\right),\]
and $\mathbf{\hat{y}}$ be the estimated vector with 
\[\hat{y}_{i,j} = \varphi \circ R\left(G_{i,j},\ u_{i,j},\ \max(0,L - \lceil \log_M (w_i + 1) \rceil) \right).\]

Define $\mathbf{x} \triangleq \varphi^{-1} ( \mathbf{y} )$ for $x_{i,j} = \varphi^{-1}(y_{i,j})$, which is applying $\varphi^{-1}$ entry-wise to $\mathbf{y}$, similarly for $\mathbf{\hat{x}} \triangleq \varphi^{-1} ( \mathbf{\hat{y}} )$.
Then
\[\varphi \circ R(G,u,L)=\varphi \circ h (\mathbf{x}) \ \ \text{ and }\ \ \varphi \circ R(G,u)=\varphi \circ h (\mathbf{\hat{x}}).\]

Denote
\[\Phi(x)\triangleq \D{\varphi(x)}{x} = \frac{1}{(1+x) \ln (1+x)}.\]

Now by Mean Value Theorem, $\exists \gamma: 0\leq \gamma \leq 1, \mathbf{\tilde{y}} =\gamma \mathbf{y} + (1-\gamma) \mathbf{\hat{y}}$ such that, let $\mathbf{\tilde{x}} \triangleq \varphi^{-1} (\mathbf{\tilde{y}}) $,
\begin{align*}
&\varphi\circ R(G,u,L) - \varphi\circ R(G,u)\\
=& \sum_{i,j} \frac{ \partial (\varphi \circ h \circ \varphi^{-1}) }{\partial y_{i,j}}\Big|_ {\mathbf{y}=\mathbf{\tilde{y}}} \cdot( \hat{y}_{i,j}- y_{i,j})  \\
=&\sum_{i,j}\left( \frac{\partial h}{\partial x_{i,j}}\Big|_{\mathbf{x}=\mathbf{\tilde{x}}}\right) \frac{\Phi(h(\mathbf{\tilde{x}}))}{\Phi(\tilde{x}_{i,j})}  \cdot ( \hat{y}_{i,j}- y_{i,j}).
\end{align*}
By induction hypothesis, we have
\[|\hat{y}_{i,j}- y_{i,j}|\leq 4\alpha^{\max(0,L - \lceil \log_M (w_i + 1) \rceil)} \leq 4\alpha^{L - \lceil \log_M (w_i + 1) \rceil}.\]

Let $\alpha = 0.9616$,
$\alpha_i = \alpha^{- \lceil \log_M (w_i+1) \rceil}$,
by substitution we have
\begin{align*}
|\varphi\circ R(G,u,L) - \varphi\circ R(G,u)|
\leq
4 \alpha^L \cdot \sum_{i,j} \abs{\frac{\partial h}{\partial \tilde{x}_{i,j}}} \frac{\Phi(h(\mathbf{\tilde{x}}))}{\Phi(\tilde{x}_{i,j})} \alpha_i.
\end{align*}
Therefore, the key is to bound the amortized decay rate defined as
\[
	{\kappa}_d(\mathbf{x}) \triangleq \sum_{i,j} \abs{\frac{\partial h}{\partial x_{i,j}}} \frac{\Phi(h)}{\Phi(x_{i,j})} \alpha_i
\]
In particular, the following Claim~\ref{prop:kappa} completes the inductive proof for (\ref{eqn:amortized-cd}).
Then condition (\ref{eqn:amortized-cd}) and Claim~\ref{prop:kappa2} implies condition (\ref{eqn:amortized-cd-all}), and concludes the proof of Lemma~\ref{lem:cd}.

\begin{claim}
	\label{prop:kappa}
	For $d\le 4$, and $\frac{1}{16} \le x_{i,j} \le 1$, ${\kappa}_d(\mathbf{x}) \le 1$.
\end{claim}

\begin{claim}
	\label{prop:kappa2}
	For $\frac{1}{16} \le x_{i,j} \le 1$,
	${\kappa}_5(\mathbf{x}) < 3$.
\end{claim}

\subsection{Choice of Potential Functions}\label{sec:cd2}

In this section, we establish Claim \ref{prop:kappa} and \ref{prop:kappa2} and thus conclude the key lemma and main theorem.
The amortized decay rate ${\kappa}_d(\mathbf{x})$ is a double summation over variables $x_{i,j}$ of two layers.
%
We first show that under our choice of the potential function, where $\varphi(x) = \ln\left( \ln (1+x) \right)$ and thus $\Phi(x)= \D{\varphi(x)}{x} = \frac{1}{(1+x) \ln (1+x)}$, the double summation can be simplified into a single summation after a suitable change of variables.
Let
\[s_i = (1+ \prod_{j=1}^{w_i} (1+x_{i,j})^{-1})^{-1}.\]
Now we have $h=\prod_{i=1}^d s_i$, and ${\kappa}_d(\mathbf{x})$ can be rewritten as
\begin{align*}
	{\kappa}_d =& \sum_{i,j} \abs{\frac{\partial h}{\partial x_{i,j}}} \frac{\Phi(h)}{\Phi(x_{i,j})} \alpha_i\\
	=& \Phi(h) h\cdot \sum_{i=1}^d\frac{\alpha_i\prod_{j=1}^{w_i} \frac{1}{1+x_{i,j}}}{1+ \prod_{j=1}^{w_i} \frac{1}{1+x_{i,j}}} \sum_{j=1}^{w_i} \abs{\frac{\frac{1}{1+x_{i,j}}}{\Phi(x_{i,j})}}\\
	=& \frac{h}{\left( 1+h\right)\ln(1+h)} \cdot \sum_{i=1}^d\frac{\alpha_i\prod_{j=1}^{w_i} \frac{1}{1+x_{i,j}}}{1+ \prod_{j=1}^{w_i} \frac{1}{1+x_{i,j}}} \sum_{j=1}^{w_i} \ln ({1+x_{i,j}})\\
	=& \frac{h}{\left( 1+h\right)\ln(1+h)} \cdot \sum_{i=1}^d\frac{\alpha_i\prod_{j=1}^{w_i} \frac{1}{1+x_{i,j}}}{1+ \prod_{j=1}^{w_i} \frac{1}{1+x_{i,j}}} \ln \left(\prod_{j=1}^{w_i}  ({1+x_{i,j}})\right)\\
	=& \frac{h}{\left( 1+h\right)\ln(1+h)} \cdot \sum_{i=1}^d \alpha_i (1-s_i) \ln\frac{s_i}{1-s_i}.
\end{align*}
Therefore, this specific potential function collapses the two-layer decay rate into a single layer one, which  only depends on $\set{s_i}$.
 As a remark, $s_i$ has a combinatorial meaning back in the original tree recursion, which corresponds exactly to an estimate of $\pr_{G_i} \left( v_i = 0 \right)$. In the following, we treat ${\kappa}_d$ as a function of $\set{s_i}$ rather than $\set{x_{i,j}}$, which significantly simplifies the proof.

Next we convert the bounds for $x_{i,j}$ to bounds for $s_i$.
Note that lowerbounding $s_i$ here is essentially by giving a lowerbound of $\pr_{G_i} \left( v_i = 0 \right)$.
%


\begin{claim}
	\label{prop:si}
	\[\frac{17^{w_i}}{16^{w_i}+17^{w_i}} \leq s_i < 1.\]
\end{claim}
As an intuition, this claim says that larger $w_i$ will only make $\pr_{G_i} \left( v_i = 0 \right)$ closer to 1. 
In other words, if we treat $v_i$ as a constraint, larger $w_i$ will only weaken the overall ``influence'' of the constraint by making it almost always satisfied, which allows the correlation to decay faster.

\begin{proof}
Since $s_i = (1+ \prod_{j=1}^{w_i} (1+x_{i,j})^{-1})^{-1}$, it is clear that $s_i<1$. As $x_{i,j}\geq \frac{1}{16}$, we have
\begin{align*}
s_i = (1+ \prod_{j=1}^{w_i} (1+x_{i,j})^{-1})^{-1} 
\geq (1+ \prod_{j=1}^{w_i} (1+\frac{1}{16})^{-1})^{-1} 
=\frac{17^{w_i}}{16^{w_i}+17^{w_i}}.
\end{align*}
%
\end{proof}

The rate ${\kappa}_d$ also involves parameters  $\alpha_i = \alpha^{- \lceil \log_M (w_i+1) \rceil}$, which are discontinuous functions in degrees $w_i$. To handle this,
we group the variables into the following two parts: 
\[
	I_1 = \set{i: w_i < M}, \ \ \text{ and } \ \ I_2 = \set{i : w_i \ge M}.
\]

	For $i\in I_1$, we have $\alpha_i=\frac{1}{\alpha}$ being a constant; for $i\in I_2$, we bound them as $\alpha_i = \alpha^{- \lceil \log_M (w_i+1) \rceil} \leq \alpha^{- \log_M (w_i+1) -1} $.
Let $d_1 = \abs{I_1}$ and  $d_2 = \abs{I_2}$, clearly we have $d_1 + d_2 = d$.

The summation in ${\kappa}_d$ is also divided into two part for $I_1$ and $I_2$.
For $i\in I_2$,
$s_i$ lies in a rather narrow range $\left[\frac{17^{w_i}}{16^{w_i}+17^{w_i}},1\right]$ with $\frac{17^{w_i}}{16^{w_i}+17^{w_i}}\geq \frac{17^{45}}{16^{45}+17^{45}} >\frac{9}{10}$.
As we will see, ${\kappa}_d$ is a decreasing function in $s_i$ for $i\in I_2$ in these ranges.
As a result, those terms corresponding to $I_2$ can be replaced by an upper bound of $\frac{1}{5}$.

For  $i\in I_1$, we use Jensen's inequality to prove that the maximum is achieved when $s_i$s are all equal to the
 same value $\hat{s}$.
 Finally, we can bound the decay rate by a function in a single variable  $\hat{s}$.
 Here is the formal definition and the proof. We define the symmetrized version of ${\kappa}_d$ as
%
%
%
%
\[ \hat{\kappa}_d(\hat{s}) = \frac{ \hat{s}^{d_1} \cdot  d_1 \cdot (1-\hat{s})\ln\left(\frac{\hat{s}}{1-\hat{s}}\right)}{\alpha \left( 2^{d_2} + \hat{s}^{d_1}\right)\ln(1+2^{-d_2} \cdot \hat{s}^{d_1})} +   \frac{d_2}{5},\]

\begin{claim}
	\label{prop:max-equiv}
	\[\max_{\mathbf{s}}{\kappa}_d \le \max_{\hat{s}} \hat{\kappa}_d .\]
\end{claim}

\begin{proof}
	We begin with some elementary inequalities.
	Let $f(x) = (1-e^x) (x - \ln(1-e^x))$,
	\[
f''(x) = -\frac{e^x \left(1+\left(1-e^x\right) \ln \left(\frac{e^x}{1-e^x}\right)\right)}{1-e^x}.
\]
Since $\frac{1}{2} \le e^x \le 1$ for $x\in [-\ln 2, 0]$, we have $f''(x) \le 0$ and $f(x)$ is concave over $x\in [-\ln 2, 0]$.
 Let $\hat{s}=\left(\prod_{i\in I_1} s_i \right)^{1/d_1}$, by Jensen's inquality, we have
\begin{align}
	\label{eqn:jensen}
	\sum_{i \in I_1} (1-s_i) \ln\frac{s_i}{1-s_i} =& \sum_{i \in I_1} f(\ln s_i) 
	\le \sum_{i \in I_1} f(\ln \hat{s}) 
	= d_1 \cdot (1-\hat{s}) \ln \frac{\hat{s}}{1-\hat{s}}.
\end{align}

Let $g(h) = \ln(1+h) - \frac{h}{1+h}$,
since $g'(h) = \frac{h}{(1+h)^2} \ge 0$ for $0\le h \le 1$,
we have $g(h) \ge g(0) = 0$, namely
\begin{align}
	\frac{h}{(1+h)\ln(1+h)} \le 1,\ \ & \text{ for $0\le h \le 1$}
	\label{eqn:log}
\end{align}

Also note that $\frac{h}{\left( 1+h\right)\ln(1+h)}$ is decreasing in $h$, and $2^{-d_2}\cdot \hat{s}^{d_1}\le h$ due to  $s_i\geq \frac{1}{2}$, we have

\begin{align*}
	{\kappa}_d \le& \frac{h}{\left( 1+h\right)\ln(1+h)}
	\left( \sum_{i\in I_1} \alpha_i (1-s_i) \ln\frac{s_i}{1-s_i}
 +\sum_{i\in I_2} \alpha_i (1-s_i) \ln\frac{s_i}{1-s_i}\right) \\
\le &\frac{2^{-d_2} \cdot \hat{s}^{d_1} \cdot }{\left( 1+2^{-d_2}\cdot \hat{s}^{d_1}\right)\ln(1+2^{-d_2}\cdot \hat{s}^{d_1})}
\sum_{i\in I_1}  \alpha^{-1}(1-s_i) \ln\frac{s_i}{1-s_i} 
+\sum_{i\in I_2} \alpha^{- \lceil \log_M (w_i + 1) \rceil} (1-s_i) \ln\frac{s_i}{1-s_i} \\
\le & \frac{2^{-d_2} \cdot \hat{s}^{d_1} d_1}{\alpha \left( 1+2^{-d_2} \cdot \hat{s}^{d_1}\right)\ln(1+2^{-d_2} \cdot \hat{s}^{d_1})} \cdot (1-\hat{s})  \ln\left(\frac{\hat{s}}{1-\hat{s}}\right) 
+\sum_{i\in I_2} \alpha^{- \lceil \log_M (w_i + 1) \rceil} (1-s_i) \ln\frac{s_i}{1-s_i}.
\end{align*}

Finally, it remains to show that, for $i\in I_2$,
\begin{align}
	\label{eqn:gamma}
	 \alpha^{- \lceil \log_M (w_i + 1) \rceil} (1-s_i) \ln\frac{s_i}{1-s_i} \le \frac{1}{5}.
\end{align}

Recall that for $i \in I_2$, $w_i \ge M=45$,
and by Claim \ref{prop:si},
we have
\[
	\frac{1}{1+\left( \frac{16}{17} \right)^{w_i}} \le s_i < 1.
\]
Also note that $(1-s) \ln\frac{s}{1-s}$ is decreasing in $s$ for $\frac{4}{5} \le s < 1$.
Let $\gamma(w) = w \left(\frac{16}{17}\right)^{w} \ln \left( \frac{17}{16} \right) \alpha^{  -\log_M (w+1) - 1}$, we have
\[
	\alpha^{- \lceil \log_M (w_i + 1) \rceil} (1-s_i) \ln\frac{s_i}{1-s_i} \le \gamma(w_i).
\]
	It can be verified that $\gamma(w)$ is a decreasing function in $w$ for $w\ge 45$,
	and as a result we have $\gamma(w) \le \gamma(45) < \frac{1}{5}$.
	Hence the relation (\ref{eqn:gamma}) follows and we conclude the proof.
\end{proof}

Now it suffices to bound the decay rate with $\hat{\kappa}_d$.
As it is a real function in a single variable $\hat{s}$ with bounded domain $\left[\frac{1}{2}, 1\right]$, there is a standard calculus method to find their maximum values and we only need to verify that they satisfy Claim \ref{prop:kappa} and \ref{prop:kappa2}. We do that on a case-by-case basis.

\bigskip
{\noindent \bf Proof of Claim \ref{prop:kappa}.}
	Recall that $\hat{\kappa}_d$ is single-variate in $\hat{s}$ with $\frac{1}{2} \le \hat{s} \le 1$, parameterized by $d_1, d_2$.

Also note that $\hat{\kappa}_d$ is increasing in both $d_1, d_2$,
whereas $d_1 + d_2 = d \le 4$. So it suffices to check that $\hat{\kappa}_4 < 1$ for each case.

	{\noindent \bf Case } $d_2 = 0$:

	In this case,  we have 
	\[\hat{\kappa}_4=\frac{4 (1-\hat{s}) \hat{s}^4 }{\alpha \left(\hat{s}^4+1\right) \ln \left(\hat{s}^4+1\right)}\ln \left(\frac{\hat{s}}{1-\hat{s}}\right).\]
	It achieves a unique maximum at $\hat{s}^*\approx 0.758669$,
	and thus $\hat{\kappa}_4(\hat{s}^*) < 1$.
	
	\bigskip
	{\noindent \bf Case } $d_2 = 1$:

	In this case,  we have 	
	\[\hat{\kappa}_4 = \frac{3 (1-\hat{s}) \hat{s}^3 }{\alpha \left(\hat{s}^3+2\right) \ln \left(\frac{1}{2} \left(\hat{s}^3+2\right)\right)}\ln \left(\frac{\hat{s}}{1-\hat{s}}\right) + \frac{1}{5}.\]
	It achieves a unique maximum at $\hat{s}^*\approx 0.7691$,
	and thus $\hat{\kappa}_4 < \frac{4}{5} + \frac{1}{5} \le 1$.

	\bigskip
	{\noindent \bf Case } $d_2 = 2$:

	In this case, we have 
	\[\hat{\kappa}_4 = \frac{2 (1-\hat{s}) \hat{s}^2 }{\alpha \left(\hat{s}^2+4\right) \ln \left(\frac{1}{4} \left(\hat{s}^2+4\right)\right)}\ln \left(\frac{\hat{s}}{1-\hat{s}}\right) + \frac{2}{5}.\]
	It achieves a unique maximum at $\hat{s}^*\approx 0.776043$,
	and thus $\hat{\kappa}_4 < 0.55 + \frac{2}{5} < 1$.

	\bigskip
	{\noindent\bf Case } $d_2 = 3$:

	In this case, we have 
	\[\hat{\kappa}_4=\frac{(1-\hat{s}) \hat{s} }{\alpha (\hat{s}+8) \ln \left(\frac{\hat{s}+8}{8}\right)}\ln \left(\frac{\hat{s}}{1-\hat{s}}\right) + \frac{3}{5}.\]
	It achieves a unique maximum at $\hat{s}^*\approx 0.780104$,
	and thus $\hat{\kappa}_4 < 0.3 + \frac{3}{5} < 1$.

	\bigskip
	{\noindent\bf Case } $d_2 = 4$:

	In this case, we have $\hat{\kappa}_4 \le \frac{4}{5} < 1$.
	\qed

	\bigskip
{\noindent \bf Proof of Claim \ref{prop:kappa2}.}
	Let $f(\hat{s}) = (1-\hat{s}) \ln \left(\frac{\hat{s}}{1-\hat{s}}\right)$.
	For $\frac{1}{2} \le \hat{s} \le 1$, $f$ achieves its unique maximum at $\hat{s}^* \approx 0.782188$, and $f(\hat{s}^*) < 0.3$.
	Hence $\hat{\kappa}_5 < \frac{5f(\hat{s}^*)}{\alpha} + 1 < 3$.
	\qed

\bibliographystyle{plain}
\bibliography{refs}

\begin{thebibliography}{10}

\bibitem{bordewich2011approximation}
Magnus Bordewich.
\newblock On the approximation complexity hierarchy.
\newblock In {\em Approximation and Online Algorithms}, pages 37--46. Springer,
  2011.

\bibitem{bordewich2006}
Magnus Bordewich, Martin Dyer, and Marek Karpinski.
\newblock Stopping times, metrics and approximate counting.
\newblock In {\em Automata, Languages and Programming}, volume 4051 of {\em
  Lecture Notes in Computer Science}, pages 108--119. Springer Berlin
  Heidelberg, 2006.

\bibitem{BulatovDGJM13}
Andrei~A. Bulatov, Martin~E. Dyer, Leslie~Ann Goldberg, Mark Jerrum, and Colin
  McQuillan.
\newblock The expressibility of functions on the boolean domain, with
  applications to counting csps.
\newblock {\em J. {ACM}}, 60(5):32, 2013.

\bibitem{cai2014bis}
Jin{-}Yi Cai, Andreas Galanis, Leslie~Ann Goldberg, Heng Guo, Mark Jerrum,
  Daniel Stefankovic, and Eric Vigoda.
\newblock {\#}{BIS}-hardness for 2-spin systems on bipartite bounded degree
  graphs in the tree non-uniqueness region.
\newblock In {\em Approximation, Randomization, and Combinatorial Optimization.
  Algorithms and Techniques, {APPROX/RANDOM} 2014, September 4-6, 2014,
  Barcelona, Spain}, pages 582--595, 2014.

\bibitem{chebolu2012complexity}
Prasad Chebolu, Leslie~Ann Goldberg, and Russell Martin.
\newblock The complexity of approximately counting stable roommate assignments.
\newblock {\em Journal of Computer and System Sciences}, 78(5):1579--1605,
  2012.

\bibitem{dyer2000relative}
Martin Dyer, Leslie~Ann Goldberg, Catherine Greenhill, and Mark Jerrum.
\newblock {\em On the relative complexity of approximate counting problems}.
\newblock Springer, 2000.

\bibitem{dyer2010approximation}
Martin Dyer, Leslie~Ann Goldberg, and Mark Jerrum.
\newblock An approximation trichotomy for boolean\# csp.
\newblock {\em Journal of Computer and System Sciences}, 76(3):267--277, 2010.

\bibitem{dyer2000markov}
Martin Dyer and Catherine Greenhill.
\newblock On {Markov chains} for independent sets.
\newblock {\em Journal of Algorithms}, 35(1):17--49, 2000.

\bibitem{IS_DFJ02}
Martin~E. Dyer, Alan~M. Frieze, and Mark Jerrum.
\newblock On counting independent sets in sparse graphs.
\newblock {\em SIAM Jounal on Computing}, 31(5):1527--1541, 2002.

\bibitem{DyerGJR12}
Martin~E. Dyer, Leslie~Ann Goldberg, Markus Jalsenius, and David Richerby.
\newblock The complexity of approximating bounded-degree boolean {\#}csp.
\newblock {\em Inf. Comput.}, 220:1--14, 2012.

\bibitem{ge2012graph}
Qi~Ge and Daniel {\v{S}}tefankovi{\v{c}}.
\newblock A graph polynomial for independent sets of bipartite graphs.
\newblock {\em Combinatorics, Probability and Computing}, 21(05):695--714,
  2012.

\bibitem{goldberg2012counterexample}
Leslie Goldberg and Mark Jerrum.
\newblock A counterexample to rapid mixing of the ge-stefankovic process.
\newblock {\em Electron. Commun. Probab.}, 17:no. 5, 1--6, 2012.

\bibitem{goldberg2007complexity}
Leslie~Ann Goldberg and Mark Jerrum.
\newblock The complexity of ferromagnetic ising with local fields.
\newblock {\em Combinatorics, Probability \& Computing}, 16(1):43--61, 2007.

\bibitem{GoldbergJ12}
Leslie~Ann Goldberg and Mark Jerrum.
\newblock Approximating the partition function of the ferromagnetic potts
  model.
\newblock {\em J. {ACM}}, 59(5):25, 2012.

\bibitem{app_JSV04}
Mark Jerrum, Alistair Sinclair, and Eric Vigoda.
\newblock A polynomial-time approximation algorithm for the permanent of a
  matrix with nonnegative entries.
\newblock {\em Journal of the ACM}, 51:671--697, July 2004.

\bibitem{jerrum1986random}
Mark~R Jerrum, Leslie~G Valiant, and Vijay~V Vazirani.
\newblock Random generation of combinatorial structures from a uniform
  distribution.
\newblock {\em Theoretical Computer Science}, 43:169--188, 1986.

\bibitem{LLY12}
Liang Li, Pinyan Lu, and Yitong Yin.
\newblock Approximate counting via correlation decay in spin systems.
\newblock In {\em Proceedings of SODA}, pages 922--940, 2012.

\bibitem{LLY13}
Liang Li, Pinyan Lu, and Yitong Yin.
\newblock Correlation decay up to uniqueness in spin systems.
\newblock In {\em Proceedings of SODA}, pages 67--84, 2013.

\bibitem{counting-edge-cover}
Chengyu Lin, Jingcheng Liu, and Pinyan Lu.
\newblock A simple fptas for counting edge covers.
\newblock In {\em Proceedings of the Twenty-Fifth Annual ACM-SIAM Symposium on
  Discrete Algorithms}, SODA '14, pages 341--348. SIAM, 2014.

\bibitem{monotone-cnf}
Jingcheng Liu and Pinyan Lu.
\newblock Fptas for counting monotone cnf.
\newblock In {\em Proceedings of the Twenty-Sixth Annual ACM-SIAM Symposium on
  Discrete Algorithms}, SODA '15, pages 1531--1548. SIAM, 2015.

\bibitem{LLZ14}
Jingcheng Liu, Pinyan Lu, and Chihao Zhang.
\newblock The complexity of ferromagnetic two-spin systems with external
  fields.
\newblock In {\em Approximation, Randomization, and Combinatorial Optimization.
  Algorithms and Techniques, {APPROX/RANDOM} 2014, September 4-6, 2014,
  Barcelona, Spain}, pages 843--856, 2014.

\bibitem{liu2014fptas}
Jingcheng Liu, Pinyan Lu, and Chihao Zhang.
\newblock {FPTAS} for counting weighted edge covers.
\newblock In {\em Algorithms-ESA 2014}, pages 654--665. Springer, 2014.

\bibitem{IS_LV97}
Michael Luby and Eric Vigoda.
\newblock Approximately counting up to four (extended abstract).
\newblock In {\em Proceedings of STOC}, pages 682--687, 1997.

\bibitem{inapp_MWW09}
Elchanan Mossel, Dror Weitz, and Nicholas Wormald.
\newblock On the hardness of sampling independent sets beyond the tree
  threshold.
\newblock {\em Probability Theory and Related Fields}, 143:401--439, 2009.

\bibitem{RestrepoSTVY11}
Ricardo Restrepo, Jinwoo Shin, Prasad Tetali, Eric Vigoda, and Linji Yang.
\newblock Improved mixing condition on the grid for counting and sampling
  independent sets.
\newblock In {\em {IEEE} 52nd Annual Symposium on Foundations of Computer
  Science, {FOCS} 2011, Palm Springs, CA, USA, October 22-25, 2011}, pages
  140--149, 2011.

\bibitem{SST}
Alistair Sinclair, Piyush Srivastava, and Marc Thurley.
\newblock Approximation algorithms for two-state anti-ferromagnetic spin
  systems on bounded degree graphs.
\newblock In {\em Proceedings of SODA}, pages 941--953, 2012.

\bibitem{sinclair2014spatial}
Alistair Sinclair, Piyush Srivastava, Daniel \v{S}tefankovi\v{c}, and Yitong
  Yin.
\newblock Spatial mixing and the connective constant: Optimal bounds.
\newblock In {\em Proceedings of the Twenty-Sixth Annual ACM-SIAM Symposium on
  Discrete Algorithms}, SODA '15, pages 1549--1563. SIAM, 2015.

\bibitem{sinclair2013spatial}
Alistair Sinclair, Piyush Srivastava, and Yitong Yin.
\newblock Spatial mixing and approximation algorithms for graphs with bounded
  connective constant.
\newblock In {\em Foundations of Computer Science (FOCS), 2013 IEEE 54th Annual
  Symposium on}, pages 300--309. IEEE, 2013.

\bibitem{Sly10}
Allan Sly.
\newblock Computational transition at the uniqueness threshold.
\newblock In {\em Proceedings of the 2010 IEEE 51st Annual Symposium on
  Foundations of Computer Science}, FOCS '10, pages 287--296, Washington, DC,
  USA, 2010. IEEE Computer Society.

\bibitem{Weitz06}
Dror Weitz.
\newblock Counting independent sets up to the tree threshold.
\newblock In {\em Proceedings of STOC}, pages 140--149, 2006.

\end{thebibliography}



\end{document}